\documentclass[letterpaper, 10 pt, conference]{ieeeconf}  % Comment this line out if you need a4paper

\IEEEoverridecommandlockouts                              % This command is only needed if 
\usepackage{subcaption}
\usepackage{balance}
\let\labelindent\relax
\usepackage{enumitem}
\usepackage{cite}
\usepackage{amsmath,amssymb,amsfonts}
\usepackage{graphicx}
\usepackage{textcomp}
\usepackage{xcolor}
\newtheorem{theorem}{\textbf{Theorem}}[section]
\newtheorem{lemma}[theorem]{\textbf{Lemma}}
\newtheorem{corollary}[theorem]{\textbf{Corollary}}
\newtheorem{definition}[theorem]{\textbf{Definition}}
\newtheorem{assumption}[theorem]{\textbf{Assumption}}

\usepackage[colorinlistoftodos]{todonotes}
\usepackage{color}
\usepackage{tcolorbox}

\newcommand{\bR} { {\mathbb R}}
\newcommand{\Ex}{\mathbb{E}}

\usepackage{mathtools} 
\DeclarePairedDelimiter{\norm}{\lVert}{\rVert}

\newcommand{\mc}[1]{\mathcal{#1}}

\usepackage{tikz}
\usetikzlibrary{arrows.meta,positioning}
\usepackage[ruled]{algorithm}
\usepackage{algpseudocode}

\algnewcommand\algorithmicinput{\textbf{Input:}}
\algnewcommand\Input{\item[\algorithmicinput]}
\algnewcommand\algorithmicoutput{\textbf{Output:}}
\algnewcommand\Output{\item[\algorithmicoutput]}
\algnewcommand\algorithmicsamples{\textbf{No. of Samples:}}
\algnewcommand\Samples{\item[\algorithmicsamples]}

\usepackage{textcomp}
\def\BibTeX{{\rm B\kern-.05em{\sc i\kern-.025em b}\kern-.08em
    T\kern-.1667em\lower.7ex\hbox{E}\kern-.125emX}}

\usepackage{optidef}

\title{\LARGE \bf
SafePG: Safe and Globally Optimal Reinforcement Learning \\with Hard Constraints
}

\author{Vipul K. Sharma, Wesley A. Suttle and S. Sivaranjani% <-this % stops a space
\thanks{Vipul K. Sharma and S. Sivaranjani are with the Edwardson School of Industrial Engineering at Purdue University, West Lafayette, IN 47907, USA.
        {\tt\small \{sharm697,sseetha\}@purdue.edu}}%}%
\thanks{Wesley A. Suttle is with U.S. Army Research Laboratory, Adelphi, MD 20783, USA.
{\tt\small
{wesley.a.suttle.ctr@army.mil}}
}}

\begin{document}

\maketitle
\thispagestyle{empty}
\pagestyle{empty}

%%%%%%%%%%%%%%%%%%%%%%%%%%%%%%%%%%%%%%%%%%%%%%%%%%%%%%%%%%%%%%%%%%%%%%%%%%%%%%%%

\begin{abstract}
We present an optimal and convergent model-free policy gradient (PG) reinforcement learning (RL) framework for controlling nonlinear dynamical systems under hard safety constraints. We first construct a class of stochastic wrapper policies centered around a deterministic controller, thereby enabling exploration in unknown environments while preserving the underlying deterministic control structure. We then define a class of parameterized safe-by-construction control policies by truncating these stochastic policies onto hard safety constraints. We next establish, via measure-theoretic arguments, that the potentially nonconvex RL objective under the truncated policy class, as well as its policy gradients, are well-defined. We then develop a model-free PG algorithm based on stochastic gradient ascent that directly searches over these truncated policies and leverage gradient dominance to establish convergence and optimality guarantees. Finally, we validate this framework through simulations on a safe quadrotor navigation problem.
\end{abstract}

\section{Introduction}
Reinforcement Learning (RL) has demonstrated remarkable success in learning control policies for complex systems by directly interacting with unknown environments. However, in safety-critical applications such as autonomous driving or  robotics, hard constraints on states and actions must be satisfied at all times, not only by the final controller but also during learning. This is particularly challenging in model-free RL, where stochastic exploration is used to improve the underlying policy. In such settings, pointwise safety requirements must hold along the realized closed-loop trajectories generated during exploration. %This requirement poses a fundamental challenge for standard model-free RL algorithms, which typically rely on unconstrained stochastic exploration strategies that can inherently drive a system into unsafe states before an effective policy is established.

In the RL literature, safety is typically addressed via constrained Markov decision processes, primal-dual methods, or other penalty-based approaches  \cite{altman2021constrained,achiam2017constrained}. While these methods incorporate safety constraints as penalties within a reward function to discourage violations in expectation, probability, or cumulative cost,  they do not in general guarantee pointwise satisfaction of hard constraints along the realized trajectories encountered during learning. In the control literature, hard constraints are often handled through model-based constrained control methods, particularly MPC, or learning-based MPC when sufficiently accurate predictive models are available or can be learned  \cite{hewing2020learning,zanon2020safe}. In learning-based control, hard constraints are also commonly enforced through projection-based mechanisms such as control barrier function (CBF)-based      ``safety filters'', where a model-free RL algorithm is used to learn a potentially unsafe controller, whose actions are subsequently projected onto safe sets \cite{brunke2022safe,ames2019control,sharma2023safe,cheng2019end, yang2025cbf}. While these projection-based filters effectively ensure safe exploration, the projection step alters the optimization landscape  and breaks the underlying optimality and convergence guarantees of the unconstrained RL algorithm \cite{suttle2024sampling}. Another line of work involves imposing hard constraints more directly through constrained neural network
optimization layers \cite{pham2018optlayer}, reduced-gradient updates \cite{ding2023reduced}, or constrained policy parameterizations \cite{suttle2024sampling}. While these methods are closer in spirit to the present setting, their  guarantees are often limited either to particular policy classes or to weaker notions of optimality.

In this paper, we address RL under hard safety constraints via policy optimization, while ensuring pointwise constraint satisfaction during exploration and establishing convergence and optimality guarantees. The main contributions of this paper are as follows: 
\begin{itemize}
    \item First, we formulate the safe RL problem with hard constraints as optimization over a truncated stochastic policy class. We propose a class of safe-by-construction policies by introducing a stochastic wrapper policy  centered around a parameterized deterministic state-feedback law and truncating its support to a given safe control set. Thus, by optimizing over this class of policies, every sampled control input satisfies hard constraints by construction, rather than through a separate projection or filtering step. Such stochastic wrapper policies have been recently employed to introduce exploration into  policy-gradient methods for deterministic continuous control settings \cite{montenegro2024learning,sharma2025globally}, but have not thus far been  employed in safe learning. The proposed framework can also accommodate a wide class of state-action constraints that are encountered in practice, including CBF-type conditions. 
    \item Second, the truncated safe policy construction alters the structure of the underlying policy optimization problem, as the resulting exploration is no longer given by the unconstrained noise processes used in standard policy-gradient analyses. Rather, the induced exploration becomes state-dependent through the closed-loop dynamics. We  use measure-theoretic arguments to establish that the resulting RL objective and its gradients are well-defined under this truncated safe policy exploration, allowing us to define valid policy optimization algorithms in this setting.
    \item  Third, we derive the score-function and policy-gradient expressions for the truncated safe policy class,  and develop a model-free policy gradient (PG) algorithm - \texttt{SafePG}, and leverage gradient dominance (Polyak-Łojasiewicz conditions), to establish convergence to a global optimum within the constrained safe policy class. This is in contrast to prior work on RL with constrained  policy parameterizations \cite{suttle2024sampling}, where the available guarantees are typically limited to convergence to stationary points. Finally, we demonstrate the proposed framework on a safe quadrotor navigation problem.
\end{itemize}
  
In summary, we believe that the proposed framework is among the first to enable pointwise satisfaction of hard safety constraints during  exploration in model-free RL while preserving rigorous policy optimization guarantees.

\subsection{Organization}
This paper is organized as follows. We first formulate the safe RL problem in Section \ref{sec:formulation} with CBF-based safety constraints for nonlinear dynamical systems. Then, in Section \ref{sec:safe-pol}, we define a new class of control policies that are safe by construction and define our model-free safe RL problem over these policies. We show that such a problem is well-defined and then provide policy gradient algorithm to update the safe policy and optimize the RL objective. We show optimality and convergence of the proposed algorithm in Section \ref{sec:opt} and validate our safe RL framework on a safe quadrotor navigation problem in Section \ref{sec:exp}. Finally, we present the proofs of all our theoretical results in the Appendix.

\subsection{Notation}
We denote the sets of real numbers and $n$-dimensional real vectors by $\mathbb{R}$ and $\mathbb{R}^n$, respectively. For a matrix $A$, $A^\top$ denotes its transpose. We use $\langle \cdot, \cdot \rangle$ and $\|\cdot\|$ to denote the standard Euclidean inner product and norm. Given $\epsilon > 0$ and $x \in \mathbb{R}^n$, $B_\epsilon(x)$ denotes the open ball of radius $\epsilon$ centered at $x$. 
A policy is said to be $\theta$-parameterized if it depends on a vector $\theta \in \Theta$, where $\Theta$ is the set of permissible policy parameters, and $\theta$-differentiable if its partial derivative, with respect to $\theta$ exists. We use $\mathbb{E}_{z_1,z_2}$ to denote expectation $\mathbb{E}_{z_1} \mathbb{E}_{z_2}$ for random variables $z_1$ and $z_2$. We define a measurable space as a pair $(X, X_\sigma)$ consisting of a 
set $X$ and a $\sigma$-algebra $X_\sigma$ on $X$. The elements of 
$X_\sigma$ are called measurable sets. Lastly, for two Lebesgue measures $\nu_1$ and $\nu_2$, we say $\nu_1 \ll \nu_2$ iff $\nu_1$ is absolutely continuous with respect to $\nu_2$.

\section{Problem Formulation}\label{sec:formulation}
We begin by formulating the model-free safe reinforcement learning problem addressed in this paper.
\subsection{Markov Decision Process (MDP)}
We consider a discounted MDP $(\mathcal{X},\mathcal{U},\mathcal{P}, c ,\gamma)$, where the underlying transition kernel denoted by $\mathcal{P}(x_{t+1}|x_t,u_t)$ is induced by the deterministic nonlinear control-affine dynamics 
\begin{equation}\label{eq:dynamics}
    % x_{t+1} =  f(x_t) + g(x_t)u_t + d_t , \ d_t \sim \mathcal{D} % \ x(0) \sim \mathcal{D},
    x_{t+1} =  f(x_t) + g(x_t)u_t % \ x(0) \sim \mathcal{D},
\end{equation}
where $x_t \in \mathcal{X} \subseteq \mathbb{R}^n$ and $u_t \in \mathcal{U} \subseteq \mathbb{R}^m$ are the state and control input respectively at time step $t$, functions $f: \mathcal{X}\to \mathbb{R}^n$, $g: \mathcal{X}\to \mathbb{R}^{n\times m}$ are locally Lipschitz over $\mathcal{X}$. {We assume that the initial state $x_0$ is sampled according to the distribution $\mc{D}$, i.e. $x_0 \sim \mc{D}$.}
% , and $d_t$ represents the unmodeled dynamics, sampled according to distribution $\mathcal{D}$. 
% This dynamics define the transition probabilities $\mathcal{P}$ as
% \begin{equation}
%     \mathcal{P}(x_{t+1}|x_t,u_t)=1,
% \end{equation}
Moreover, $c(x_t,u_t)$ is the cost at time step $t$, and $\gamma \in [0,1]$ is the discount factor. Throughout the paper, we assume that the dynamics and hence the transition probabilities of the MDP are unknown.
% the long-horizon reward function is given by
% \begin{equation}
%     J(x_0) = \sum_{t=0}^{\infty} \gamma^t r(x_t,u_t).
% \end{equation}

\subsection{Safety}
% Our notion of safety is concerned with satisfaction of time-varying state-action coupled constraints, given by
% \begin{equation}
    
% \end{equation}
Our notion of safety is concerned with keeping the system inside a prescribed subset of the state-space $\mc{X}$, called the safe set. We define this safe set as a zero-super-level of a continuously differentiable function $h$, given by
\begin{equation}\label{eq:safe_set}
    \mathcal{S} = \{x \in \mathcal{X} : h(x) \geq 0\}.
\end{equation}
Then, safety can be guaranteed, under the dynamics \eqref{eq:dynamics} by imposing a forward invariance condition given by %choosing the a control input $u_t$ such that
\begin{equation}\label{eq:orig-cbf}
    h(x_{t+1}) - (1-\alpha) h(x_t) \geq 0.
\end{equation}
A continuously differentiable function $h$ satisfying this condition is termed a control barrier function (CBF). The above condition induces the set of safe control inputs at time step $t$, termed the \textit{safe control set}, given by
\begin{equation} 
    C(x_t) = \{u_t \in \mathcal{U}: h(x_{t+1}) - (1-\alpha) h(x_t) \geq 0\}.
\end{equation}

\subsection{Safe Reinforcement Learning (RL) Problem}
We now state the safe RL problem under hard safety constraints over the MDP in \eqref{eq:dynamics}, with $\gamma\in [0,1]$, given the point-wise hard safety constraint $u_t \in C(x_t)$ for all $t \geq 0$. We solve the following optimization problem to find an optimal policy $\pi$, with $u_t \sim \pi(\cdot|x_t)$, that solves
%
% \chck{we need to say something like $u \sim pi$?}\vks{check}
%
\begin{mini}
    {\pi}{ J :=   \mathbb{E}_{x_0,  \pi}\sum_{t=0}^\infty \gamma^t c(x_t,u_t)  }{}{} \label{opt:constrained_lqr_control}
    %
    % \addConstraint{x_{t+1} }{= f(x_t)+}{ B u_t, }{\quad x_0 \sim \mathcal{D}}{}
    %
    \addConstraint{u_t \in C(x_t),} {} {\quad  x_{t+1} =  f(x_t) + g(x_t)u_t}{ } {},
\end{mini}
% \begin{mini}
%     {}{ J(x_0) := \mathbb{E} \left[ \sum_{t=0}^\infty \gamma^t r(x_t,u_t) \right] }{}{} \label{opt:constrained_lqr_control}
%     %
%     % \addConstraint{x_{t+1} }{= f(x_t)+}{ B u_t, }{\quad x_0 \sim \mathcal{D}}{}
%     %
%     \addConstraint{u_t \in C(x_t),} {} {\quad  x_{t+1} =  f(x_t) + g(x_t)u_t}{ } {},
% \end{mini}
without knowledge of the MDP transition probabilities. 
\section{Safe Policy Optimization \\with Hard Constraints}\label{sec:safe-pol}
In this section, we define a class of truncated safe policies that guarantees the safety constraints by construction, and then redefine the RL objective over these safe-by-construction policies. Thereafter, we analyze the well-definedness of this new objective, paving the way to justify using a policy gradient algorithm to update these safe policies. Finally, we obtain  model-free policy gradient expressions, and then provide the resultant model-free policy gradient algorithm for our model-free safe RL problem.

\subsection{Safe Exploration with Deterministic Policies}
Control tasks often employ a deterministic policy, denoted by $\mu_\theta$ with $\theta \in \Theta$. For example, $\mu_\theta$ can be a linear state-feedback policy $\mu_{\theta}(x_t)=\theta x_t$. More generally, $\mu_{\theta}$ may denote any differentiable deterministic policy parameterization, including neural networks. 
% , such as PID control with $\theta =\begin{bmatrix}
%     K_P & K_I & K_D
% \end{bmatrix}^\top$. 
When the dynamics is unknown, $\mu_\theta$ can be learned using exploration while interacting with the environment in an RL framework. To enable exploration, we add a  zero-mean Gaussian noise $\epsilon_t \sim \mc{N}(\textbf{0}_{m}, \Sigma_t)$ to the deterministic $\mu_\theta$ to obtain
    \begin{equation}\label{eq:unsafe_pi}
        \pi_\theta(\cdot|x_t) \sim \mc{N}(\mu_\theta(x_t),\Sigma_t),
    \end{equation}
    where $\Sigma_t \in \bR^{m \times m}$ is the covariance matrix. We refer to $\epsilon_t$ as the {exploration noise}.
\begin{assumption}\label{asum:mu_diff}
    We assume that the $\theta$-parametrized deterministic policy is $\theta$-differentiable.
\end{assumption}
% \subsection{Safe Policy via Truncation}
% In this subsection, we define a class of policies that satisfy the safety constraints by construction. Consider a class of $\theta-$parameterized, differentiable stochastic policies $\pi_\theta(\cdot|x_t)$ over a compact parameter space $\Theta$, with $\theta \in \Theta$. 
As we see, adding a Gaussian exploration noise to the $\mu_\theta$ yields an infinite-support stochastic policy $\pi_\theta$, which helps us obtain a class of $\theta$-parameterized safe policies via truncation. 

We define a class of truncated safe policies by restricting the support of $\pi_\theta$ to the safe control set $C(x_t)$, given by
\begin{equation} \label{eq:piC_definition}
    \pi^C_\theta(u | x_t) =
    \begin{cases}
        \frac{\pi_\theta(u | x_t)}{\pi_\theta (C(x_t) | x_t)} & u \in C(x_t) \\
        0 & u \notin C(x_t),
    \end{cases}
\end{equation}
where $\pi_\theta(C(x_t) | x_t) = \int_{C(x_t)} \pi_\theta(u | x_t) du$. Note that \eqref{eq:piC_definition} constitutes a safe-by-construction policy class, and optimizing over \eqref{eq:piC_definition} automatically guarantees $u_t \in C(x_t)$. 
% \begin{lemma}\label{lem:piC_diff}
%     Suppose that Assumption \ref{asum:mu_diff} holds. Then, the class of truncated safe policies $\pi_\theta^C(\cdot|x)$ are $\theta-$differentiable, for every $x_t$ and $u_t \in C(x_t)$, and it is given by $\nabla_\theta \pi_\theta^C(u_t|x_t) =$
%     \begin{equation}
%          \pi_\theta(u_t|x_t) \left[(\frac{\partial\mu_\theta(x_t)}{\partial \theta})^\top (\Sigma_t^{-1})^\top [u_t-\mu_\theta(x_t)]\right].
%     \end{equation}
% \end{lemma}
% \begin{remark}\label{rem:deter-stoch}
%     For control tasks, a deterministic policy, denoted by $\mu_\theta$ is often a preference, such as PID control. However, in model-free setting where the dynamics is unknown, $\pi_\theta$ can be learned using exploration while interacting with the environment in an RL framework. In such cases, we can add a zero-mean white noise $\epsilon_t \sim \mc{N}(\textbf{0}_{m\times m}, \Sigma^2_t)$ to $\mu_\theta$ to obtain
%     \begin{equation}
%         \pi_\theta(\cdot|x_t) \sim \mc{N}(\mu_\theta,\Sigma^2_t),
%     \end{equation}
%     where $\Sigma^2_t$ is the covariance matrix. We refer to $\epsilon_t$ as the \textbf{exploration noise}.
% \end{remark}
%Now, we define a model-free RL formulation over this class of safe policies.
% \subsection{Model-free Safe RL Objective}
We define the following infinite horizon discounted objective over the class of safe policies $\pi^C_{\theta}$, with $u_t \sim \pi^C_{\theta}(\cdot|x_t)$, as
\begin{equation} \label{eq:obj_fn}
    J_\theta^C = \mathbb{E}_{x_0,\pi_{\theta}^C} \left[ \sum_{t=0}^{\infty} \gamma^t c(x_t, u_t)  \right].
\end{equation}
Then, \eqref{eq:obj_fn} is an unconstrained problem where we minimize the RL objective by searching over the class of safe-by-construction policies,  thus guaranteeing a safe exploration  during learning. 

\subsection{Well-Definedness of the Safe RL Objective}
Our goal is to develop a policy gradient algorithm that solves \eqref{eq:obj_fn}.
%, without any information on system dynamics $\mc{M}(x,u)$. 
%
To use any policy gradient algorithm, we first need to establish that the objective in \eqref{eq:obj_fn} is well-defined under the truncated safe policy class $\pi_\theta^C$. We begin by discussing reachability of the safe set $\mc{S}$ with control inputs that satisfy the forward invariance condition, i.e. $u \in C(x), x \in \mc{S}$.
%Based on the above definition and forward invariance fact$\mathcal{S}$, we move to identify specific conditions under which $\pi_{\theta}^C-$induced Markov chain is $\mu-$irreducible on $\mathcal{S}$. 
Suppose that the one-step reachable subset $W(x) \subset \mathcal{S}$ under the dynamics \eqref{eq:dynamics}, for  $x \in \mc{S}$, is given by
\begin{equation}
    W(x) = \{z \in \mathcal{S} | z= f(x) + g(x) u , u \in C(x)\}.
\end{equation}
That is, $W(x)$ contains all the reachable states in one time step when starting from a state $x$ in the set $\mc{S}$. Furthermore, for any subset $\mc{A} \subset \mc{S}$, we define its reachable subset as $R(\mc{A}) = \cup_{x \in \mc{A}} W(x)$, which implies that the reachable set of the subset $\mc{A}$ is fully covered by the reachable sets of $x \in \mc{A}$. Finally, we define the set of safe control inputs that drive any $x \in \mc{S}$ into the subset $\mc{A} \subset \mc{S}$ as
\begin{equation}
    \mc{M}^{-1}_x (\mathcal{A}) := \{ u \in C(x) \ | \ f(x) + g(x) u \in \mathcal{A} \}.
\end{equation}
Based on the above definitions, we now define the measure of the reachable set and its pushforward  on the safe set under the dynamics \eqref{eq:dynamics}. 
% consider $\ell$ as 
% We begin with the definition of the pushforward measure. 
%\begin{definition}
{
Suppose $X_1 = \mc{M}^{-1}_x (\mathcal{S})$
%Suppose $X_1 = \cup_{x \in \mc{S}} C(x)$
and $X_2 =\mc{S}$, and that $X_{1\sigma}$, $X_{2\sigma}$ represent their corresponding $\sigma$-algebras, respectively. Furthermore, suppose that $\ell: X_{1\sigma} \rightarrow [0, +\infty]$ is the Lebesgue measure on $X_1$ and define the map $\mc{T}:=T^x(u)=f(x)+g(x)u$. With a slight abuse of notation, we suppress the dependence of $\mc{T}$ on $x$ in what follows.}
    Then, with measurable spaces $(X_1, X_{1\sigma})$ and $(X_2, X_{2\sigma})$, the map $\mc{T} \colon X_1 \to X_2$, and the measure
    $\ell \colon X_{1\sigma} \to [0, +\infty]$
    , the \textbf{pushforward} of $\ell$ by $\mc{T}$ is the measure $\mc{T}_{\#}(\ell) \colon X_{2\sigma} \to [0, +\infty]$ defined as:
    \begin{equation}
        (\mc{T}_{\#}\ell)(G) = \ell\left(\mc{T}^{-1}(G)\right) \quad \text{for all } G \in X_{2\sigma}
    \end{equation}
 % \[
 % (\mc{T}_*\ell)(G) = \ell\left(\mc{T}^{-1}(G)\right) \quad \text{for all } B \in \Sigma_2
 % \]
 where $\mc{T}^{-1}(G) = \{u \in X_1 \mid \mc{T}(u) \in G\}$ is the pre-image of $G$ under $\mc{T}$.
% \end{definition}
% \chck{Suppose $\ell$ denotes the Lebesgue measure on the union of safe control sets $\cup_{x \in \mc{S}} C(x)$. Furthermore, for any given $x \in \mc{X}$, suppose the dynamics is written as a map $T^x(u) = f(x) + g(x) u$ that maps a given $u \in \mc{U}$ to  $x' \in \mc{X}$ such that $x'=f(x)+g(x)u$. Then, for any $\mc{G} \subset \cup_{x \in \mc{S}} C(x)$, we define $\mu:=T_{*}^x(\ell)(\mc{G})$ as the pushforward measure \cite{folland1999real}. Then, we define $\mu$-irreducibility as,}
Then, we define $\mc{T}_{\#}\ell$-irreducibility as follows.
% Suppose $\mu$ denote the Lebesgue measure defined over state-space $\mc{X}$, then we define $\mu$-irreducibility as,
%
\begin{definition}
A $\pi^C_{\theta}$-induced Markov chain $\{x_k\}_{k \in \mathbb{N}}$ is $\mc{T}_{\#}\ell$-irreducible on $\mathcal{S}$ if, for any $\mc{T}_{\#} \ell$-measurable $\mathcal{B} \subset \mathcal{S}$, 
\begin{equation}
    \mc{T}_{\#} \ell(\mathcal{B}) > 0 \implies \sum_{k \in \mathbb{N}} \mathbb{P}(x_k \in \mathcal{B} \ | \ x_0 = x) > 0, \forall x \in \mathcal{S}.
\end{equation}
\end{definition}
This definition connects system reachability and Markov chain irreducibility, implying that a Markov chain is irreducible on $\mathcal{S}$  if every measurable subset of $\mathcal{S}$ with positive volume  is reachable from any initial state $x_0 \in \mathcal{S}$ with positive probability. In fact, any Markov chain $\{ x_k \}_{k \in \mathbb{N}}$ induced by policy $\pi_\theta^C$ will belong to the safe set $\mathcal{S}$, following the fact that every input $u_k \sim \pi_\theta^C(\cdot|x_k)$, with $x_k \in \mathcal{S}, k\geq 0$, guarantees that the state evolution $x_{k+1} \in \mathcal{S}$.

We now make some assumptions on system reachability that are important to ensure that the safe set is reachable under the dynamics \eqref{eq:dynamics}, and that we have well-defined measures on the reachable subsets.
\begin{assumption} \label{asum:positive_volume_reachability}
    % For each $x \in \mathcal{S}$, $\mc{T}_{\#} \ell(W(x)) > 0$, and, for every measurable set $\mathcal{B} \subset \mathcal{S}$, there exists $n \in \mathbb{N}$ such that $\mathcal{B}$ is reachable in $n$ steps from $x$.
    {For each $x \in \mathcal{S}$, $\mc{T}_{\#} \ell(W(x)) > 0$. Moreover, for every measurable set $\mathcal{B} \subset \mathcal{S}$ and every $x \in \mathcal{S}$, there exists a finite integer $s(x,\mathcal{B}) \in \mathbb{N}$ such that $\mathcal{B}$ is reachable from $x$ in $s$ steps.}
\end{assumption}
This assumption implies that (i) the reachable set $W(x)$ with $x \in \mc{S}$ has strictly positive volume; (ii) any safe subset  $\mc{B} \subset \mathcal{S}$ is reachable in at most $s$ steps from $x \in \mc{S}$. Furthermore, we make the following assumption on measure of reachable subsets and their corresponding set of control values.
\begin{assumption} \label{asum:volume_preservation}
    {For any $x \in \mathcal{S}$ and any $\mc{T}_{\#} \ell$-measurable set $\mathcal{A} \subset W(x)$, $\mc{T}_{\#} \ell(\mathcal{A}) > 0$ iff $\ell(\mc{M}^{-1}_x(\mathcal{A})) > 0$.}
        % For any $x \in \mathcal{S}$ and any $\mu$-measurable set $\mathcal{A} \subset W(x)$, $\mc{T}_{\#} \ell(\mathcal{A}) > 0$ iff $\mc{T}_{\#} \ell(\mc{M}^{-1}_x(\mathcal{A})) > 0$.
\end{assumption}
This assumption ensures 
that the dynamics (1) maps positive measure subsets of control inputs to positive measure subsets of the state space and vice versa. This assumption enables application of Lebesgue-Radon-Nikodym Theorem \cite[\S3.2]{folland1999real} in our irreducibility results. %Assumption \ref{asum:positive_volume_reachability} is strongly related to reachability and controllability of the underlying nonlinear dynamical systems \eqref{eq:dynamics}.

Given that we are searching for the optimal safe policy over the class of truncated safe policies $\pi^C_\theta$, we make the following assumption on this policy class that enables sampling safe inputs.
\begin{assumption} \label{asum:positive_probability}
{
    For any subset $\mathcal{F} \subset C(x), x \in \mathcal{S}$ with positive measure $\ell(\mathcal{F}) > 0$, the policy $\pi^C_{\theta}(\cdot | x), \theta \in \Theta$ assigns positive probability to $\mathcal{F}$, that is, $\int_\mathcal{F} \pi^C_{\theta}(a | x) da > 0$.}
    % For any subset $\mathcal{F} \subset C(x), x \in \mathcal{S}$ with positive measure, i.e. $\mc{T}_{\#} \ell(\mathcal{F}) > 0$, the policy $\pi^C_{\theta}(\cdot | x), \theta \in \Theta$ assigns positive probability to $\mathcal{F}$, i.e. $\int_\mathcal{F} \pi^C_{\theta}(a | x) da > 0$.
\end{assumption}
This is a mild support condition, common in policy optimization, and ensures that the truncated policy assigns a positive probability to every positive-measure subset of the safe control set.
% Such an assumption is common in RL literature, ensuring a safe control input is sampled with a strictly safe control set has a positive probability.
%
%

%The above assumptions on the dynamics and truncated policy class  enable us to provide the following result that 
{Under the preceding assumptions, irreducibility on the forward-invariant safe set ensures that the expectation defining \eqref{eq:obj_fn}, and hence the objective, are well-defined under the truncated policy class.} 
% We begin by showing that the induced MDP, obtained by the safe policies \eqref{eq:piC_definition} applied to the deterministic dynamics \eqref{eq:dynamics}, is 
% $\mu$-irreducible on the safe set \eqref{eq:safe_set}, under suitable assumptions. 
\begin{theorem} \label{thm:irreducibility}
    Suppose that Assumptions \ref{asum:positive_volume_reachability}, \ref{asum:volume_preservation} and \ref{asum:positive_probability} hold. 
    Further, for a given $\theta \in \Theta$, suppose $\{ x_n \}$ is a Markov chain induced by the policy $\pi_\theta^C$ with dynamics \eqref{eq:dynamics}. Then $\{ x_n \}$ is $\mc{T}_{\#} \ell$-irreducible on the safe set $\mathcal{S}$.
    % Then, for a given $\theta \in \Theta$ and any subset $\mathcal{A} \subset \mathcal{S}$ such that $\mc{T}_{\#} \ell(\mathcal{A}) > 0$, the Markov chain induced by $\pi^C_{\theta}$ on $\mathcal{S}$ enters $\mathcal{A}$ with positive probability.
\end{theorem}
This irreducibility enables us to conclude the following well-definedness result of the safe RL problem.
\begin{corollary}\label{corol:well-defined}
    The objective in \eqref{eq:obj_fn} is well-defined.
\end{corollary}
 Now that we have established that the objective \eqref{eq:obj_fn} is well defined, we can use its gradient with respect to $\theta$, that is, the policy gradient, to design a descent algorithm and find the optimal safe policy. 
 %
 %See \cite[\S2.3]{konda2002thesis} for details on irreducibility in this setting.
% Now that we are assured that the objective function \eqref{eqn:discounted_obj_fn} is well-defined, we are justified in attempting to perform gradient ascent on it. In order to accomplish this, however, we need access to gradient estimates. This is the subject of the next section.
%
\subsection{Policy Gradient (PG) over Safe Policies}
We now compute the score functions for the truncated safe-by-construction policy class and build a safe policy gradient algorithm to solve the optimization problem \eqref{eq:obj_fn}. Using standard model-free policy gradient theorems \cite{konda2002thesis,sutton1999policy}, we write the following  gradient expressions:
\begin{equation}
\nabla_\theta J_\theta^C = s_\gamma \Ex_{\pi_\theta^C} \left[Q^{\pi_\theta^C} \nabla_\theta \log \pi_\theta^C(u|x) \right],
\end{equation}
where
$s_\gamma   = \begin{cases} \frac{1}{1-\gamma}  &, \gamma \in [0,1),\\  1 &, \gamma =1 \end{cases}$, with $\gamma \in [0,1)$ and $\gamma=1$ corresponding to the discounted and undiscounted cases respectively. As $\pi_\theta^C$ is obtained via truncation of the unsafe policy $\pi_\theta$ which has a Gaussian structure $\pi_\theta(\cdot|x_t) \sim \mc{N}(\mu_{\theta}(x_t),\Sigma_t)$ from \eqref{eq:unsafe_pi},  we can write the following.  
\begin{lemma}[Score function of the truncated safe policy]
\label{lem:score}
Suppose that Assumption \ref{asum:mu_diff} holds. Then, the score function of the truncated policy $\pi_\theta^C$ is given by
    \begin{equation}\label{eq:score}
        \nabla \log \pi^C_{\theta}(u_t | x_t)  =  (\frac{\partial\mu_\theta(x_t)}{\partial \theta})^\top (\Sigma_t^{-1})^\top [u_t-m_C(x_t)],
    \end{equation}
    where
    \begin{equation}\label{eq:m_H}
        m_C(x_t) = \frac{\int_{C(x_t)} u \ \pi_{\theta}(u | x_t) du}{\pi_{\theta}(C(x_t) | x)}.
    \end{equation}
\end{lemma} 

\subsection{Algorithm}
{
We estimate the score function by first estimating the function $m^C(x_t)$ in \eqref{eq:m_H} using Monte Carlo with $M$ samples $\{u_i\}_{i=1}^M$ drawn uniformly from $C(x_t)$. Using the identity
\[
\int_{C(x_t)} \varphi(u)\,du
=
\ell(C(x_t))\,\mathbb E_{u\sim \mathrm{Unif}(C(x_t))}[\varphi(u)],
\]
we obtain
\begin{align*}
\pi_\theta(C(x_t)\mid x_t)
&=
\int_{C(x_t)} \pi_\theta(u\mid x_t)\,du\\
&\approx
\frac{\ell(C(x_t))}{M}\sum_{i=1}^M \pi_\theta(u_i\mid x_t),
\end{align*}
and similarly,
\[
\int_{C(x_t)} u\,\pi_\theta(u\mid x_t)\,du
\approx
\frac{\ell(C(x_t))}{M}\sum_{i=1}^M u_i\,\pi_\theta(u_i\mid x_t).
\]
Therefore,
\begin{equation}\label{eq:est_m_H}
\hat m^C(x_t)
=
\frac{\sum_{i=1}^M u_i\,\pi_\theta(u_i\mid x_t)}
{\sum_{i=1}^M \pi_\theta(u_i\mid x_t)}.
\end{equation}}
We now describe the safe policy gradient method \texttt{SafePG}, following {Flowchart \ref{fig:flowchart}}, to update the policy parameters $\theta$ of the safe-by-construction policy $\pi_\theta^C$ in Algorithm \ref{alg:safePG}.

\begin{algorithm}[htb]
\caption{\texttt{SafePG:} Provably Safe Policy Gradient Reinforcement Learning 
% \chck{edit to costs and grad descent, also several notational inconsistencies} \vks{check now}
}
\label{alg:safePG}
\begin{algorithmic}[1]
    \Input $x_0, \theta_0$, Monte carlo sample size $M$, rollout length $N$, discount factor $\gamma$
    \Output Safe and globally optimal policy
    \Statex % Replaces \BlankLine for a small vertical space
    \State \textbf{Initialization:} Set $t \gets 0$.
    \Repeat
        \State Construct $\pi_{\theta_t}(\cdot|x_t)$ as $\pi_{\theta_t}(\cdot|x_t)\sim \mc{N}(\mu_{\theta_t}(x_t),\Sigma_t)$
        \State Construct $\pi^C_{\theta_t}(\cdot | x_t)$ with $\pi_{\theta_t}(\cdot|x_t)$, $C(x_t)$
        \State  $u_t \sim \pi^C_{\theta_t}(\cdot | x_t)$
        %\State Compute per step cost $r(x_t,u_t)$
        \For{$k = t, \ldots, t+N-1$}
            
            % \State Compute $r(x_k,u_k)$
            % \State $\hat{Q} \gets \hat{Q} + \doubt{\gamma^{t / 2}} r(x_k, u_k)$
            \State $x_{k+1} \sim \mathcal{P}(\cdot | x_k, u_k)$
            % \State Construct $\pi^C_{\theta_t}(\cdot | x_{k+1})$ with $\pi_{\theta_t}(\cdot|x_{k+1})$, $C(x_{k+1})$
            \State $u_{k+1} \sim \pi^C_{\theta_t}(\cdot | x_{k+1})$
            \State Compute per-step cost $c(x_{k},u_{k})$
        \EndFor
        \State Compute $\widehat{Q^{\pi_{\theta_t}^C}} = \sum_{i=t}^{t+N-1} \gamma^{i-t}  c(x_{i}, u_{i})$
        % \State \hat{Q}^{\pi^C_{\theta_k}} ( x_{T_{k+1}}, u_{T_{k+1}} ) = \hat{Q}
        % \State $\hat{Q}^{\pi^C_{\theta_k}}(x_{T_{k+1}}, u_{T_{k+1}}) = \hat{Q}$
        % \State Uniformly sample $\{ u_l \}_{l=1}^M$ from $C(x_t)$, to compute $\widehat{\nabla \log \pi^C_{\theta_t}}(u_t | x_t)$
        \State Compute $\widehat{m_C(x_t)}$ using \eqref{eq:est_m_H} with $M$
        \State Compute $\widehat{\nabla \log \pi^C_{\theta_t}}(u_{t} | x_t)$ using \eqref{eq:score} with $\widehat{m_C(x_t)}$
        % \State Compute $\widehat{\nabla \log \pi^C_{\theta_t}}(u_t | x_t)=\frac{\sum_{i=1}^N u_i \pi_{\theta_t}(u_i|x_t)}{\sum_{j=1}^N \pi_{\theta_t} (u_i|x_t)}$
        \State Compute $\widehat{\nabla_\theta J_{\theta_t}^C} = s_\gamma \widehat{Q^{\pi_{\theta_t}^C}} \widehat{\nabla \log \pi^C_{\theta_t}}(u_{t} | x_{t})$
        % \State $\theta_{k+1} \gets \theta_k + \frac{ \alpha_k }{ 1 - \gamma } \hat{Q}^{\pi^C_{\theta_k}} ( x_{T_{k+1}}, u_{T_{k+1}} ) \widehat{\nabla \log \pi^C_{\theta_k}}(u_{T_{k+1}} | x_{T_{k+1}})$
        \State $\theta_{t+1} \gets \theta_t -  \eta_t \widehat{\nabla_\theta J_{\theta_t}^C}$
        \State $t \gets t + 1$
    \Until{convergence}
\end{algorithmic}
\end{algorithm}

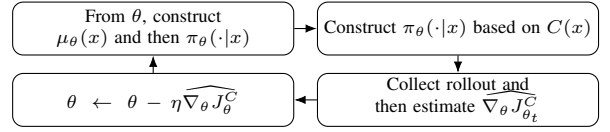
\begin{figure}[t]
\centering
\begin{tikzpicture}[
node distance=2.3mm and 3mm,
box/.style={draw,rounded corners,align=center,font=\scriptsize,minimum height=7mm,text width=.42\columnwidth,inner sep=1.3pt},
arr/.style={-{Latex[length=1.5mm]},thin}]
\node[box] (a) {From $\theta$, construct $\mu_\theta(x)$ and then $\pi_\theta(\cdot|x)$};
\node[box,right=of a] (b) {Construct $\pi_\theta(\cdot|x)$ based on $C(x)$};
\node[box,below=of b] (c) {Collect rollout and then estimate $\widehat{\nabla_\theta J_{\theta_t}^C}$};
\node[box,left=of c] (d) {$\theta \gets \theta -  \eta \widehat{\nabla_\theta J_{\theta}^C}$};
\draw[arr] (a)--(b); \draw[arr] (b)--(c); \draw[arr] (c)--(d); \draw[arr] (d)--(a);
\end{tikzpicture}
\caption{Flowchart for updating the policy parameter $\theta$}
\label{fig:flowchart}
\end{figure}
We will next prove convergence and optimality of this algorithm in the following section. 
\section{Optimality and Convergence}\label{sec:opt}
In this section, we characterize the conditions for convergence and global optimality of the \texttt{SafePG} algorithm. 
% We begin with the weak gra, first introduced in \cite{montenegro2024learning}. 
%
% \begin{assumption}
%     We assume that the deterministic policy $\pi_\theta^C$ is Lipschitz and sufficiently smooth.
% \end{assumption}
%
% \begin{theorem}
%     We show that the induced MDP is log-Lipschitz.
% \end{theorem}
%
% \begin{theorem}
%     We show the induced MDP is smooth.
% \end{theorem}
% Suppose that the unconstrained RL problem under policy $\pi_\theta$ is defined as 
% \begin{equation}
%     J_{\theta}:=\Ex_{\pi_\theta} \sum_{t=0}^\infty  \gamma^t r(x_t,u_t).
% \end{equation}
% $J_{\theta}:=\Ex_{\pi_\theta} \sum_{t=0}^\infty  \gamma^t r(x_t,u_t)$.
%
We begin by stating a structural assumption used to establish convergence and optimality of Algorithm \ref{alg:safePG}, namely a Polyak-Lojasiewicz (PL) condition, also referred to as gradient dominance (GD), on the objective \eqref{eq:obj_fn}.
\begin{assumption}[Weak Gradient Dominance]\label{asum:gd}
% \chck{Why do we care about the global optima of the unconstrained objective? isn't the PL condition (17) on the constrained objective? can we remove the first sentence here?} \vks{This sentence was part of the previous version of gradient dominance. In the current instance, we do not need to mention anything about the unconstrained problem.}
    Suppose that the global optima of 
    % $J_\theta$, i.e. the RL objective without hard constraints, is given by $J_\theta^{*}$. Then, we assume that
    the hard constrained RL objective $J_\theta^C$ from \eqref{eq:obj_fn} is given by $J_\theta^{C^*}$. Then, $J_\theta^C$ satisfies the following weak gradient dominance property, with $\alpha, \beta>0$, as\begin{equation}\label{eq:wgd}
        J_\theta^C - J_\theta^{C^*} \leq \alpha \norm{ \nabla_\theta J_\theta^C } + \beta,
    \end{equation}
    %\norm{J_\theta^C - J_\theta^*}
\end{assumption}
This condition bounds suboptimality in terms of the gradient norm, up to the additive term $\beta$, and is a key requirement in establishing global optimality guarantees for nonconvex problems using first-order methods. A strong version of this condition has been shown to hold for unconstrained LQR \cite{fazel2018global}, while a weak version akin to Assumption \ref{asum:gd} has been established for the proportional-integral-derivative (PID) control parameterization \cite{sharma2025globally}. We next state a standard regularity assumption on the policy gradient.
% This condition provides an upper bound on `how far an objective is from its optima' using the gradient norm, and is a key requirement when analyzing global optimality for non-convex problems with first order optimization methods. A strong version of this condition has been shown to be satisfied for LQR \cite{fazel2018global}, whereas a weaker version holds for optimal proportional-integral-derivative (PID) control \cite{sharma2025globally}. With this, we present another assumption, essential for PG algorithms.
\begin{assumption}[Lipschitzness of the Policy Gradient]\label{asum:lip-grad}
    The policy gradient $\nabla_\theta J_\theta^C (x_0)$ is Lipschitz, satisfying
    \begin{equation}\label{eq:lip-grad}
        \norm{\nabla_\theta J_{\theta'}^C - \nabla_\theta J_{\theta}^C} \leq L \norm{\theta' - \theta}.
    \end{equation}
\end{assumption}
% \chck{Does Malik show global or local lipschitzness?  same with PID?} \vks{both Lipschitzness results are local. Wes insisted that we keep gloabl assumption..similar to AISTATS and standard RL works}
This is a standard regularity assumption in policy-gradient analysis, ensuring that the gradient varies smoothly with the policy parameter. 
% In control settings, this condition has been shown to hold for control parameterizations like LQR \cite{malik2020derivative} and PID  \cite{sharma2025globally}.
% For control frameworks, this assumption  holds for both LQR \cite{malik2020derivative} and Optimal PID \cite{sharma2025globally}.
 A final ingredient needed to analyze the convergence of Algorithm \ref{alg:safePG} is a standard assumption on the stochastic gradient estimator used in the update. Since the policy gradient is estimated from finite-horizon rollouts, the resulting estimator is subject to sampling noise. We therefore assume that this estimator is unbiased and has variance that decreases with the batch size. This assumption is standard in policy-gradient analysis and allows us to control the effect of sampling noise in the convergence analysis.
\begin{assumption}\label{asum:zero-bdd-grad}
    The gradient estimator $\widehat{\nabla_\theta J_{\theta}^C}$, computed with batch size $N>0$ in Algorithm \ref{alg:safePG}, is unbiased and has a bounded variance for some bounded parameter $V$, given as  
    \begin{equation}
        \Ex \left[\widehat{ \nabla_\theta J_{\theta}^C}\right] =  \nabla_\theta J_{\theta}^C,
        \text{Var}\left[ \widehat{\nabla_\theta J_{\theta}^C} \right] \leq \frac{V}{N}.
    \end{equation}
\end{assumption}

% Such gradient dominance holds true for standard LQR problem \cite{fazel2018global} with $\beta=0$, and optimal PID control problem \cite{sharma2025globally}.
% and is a typical assumption found in works such \cite{montenegro2024learning}.

% Such an assumption is typical in many RL works, such as \cite{montenegro2024learning} and \cite{sharma2025globally}.

% \begin{theorem}[Global Optimality]
%     Suppose GD holds with $\beta=0$. Then convergence to global optimality, under hard constraints, holds. 
% \end{theorem}

% \begin{theorem}
%     Suppose that Assumption \ref{asum:gd} holds for $J_\theta^C$. Then convergence to $\beta-$optimality, under hard constraints, holds.
% \end{theorem}
We are now ready to state the main convergence and sample-complexity result for \texttt{SafePG}.
\begin{theorem}\label{thm:disc-samp-comp}
    Suppose Assumptions \ref{asum:gd}, \ref{asum:lip-grad} and \ref{asum:zero-bdd-grad} hold. Further, suppose that the safe RL algorithm \ref{alg:safePG} runs for $T>0$ iterations with updates $\{\theta_i\}_{i=1}^T$, rollout length $N>0$ and step-size $\zeta$ satisfying 
    % \chck{should it be $c_T$ here? can we use a different notation as $c$ is the per-stage cost?} \vks{changed}
    \begin{equation}
        \zeta \leq \min \{ \frac{1}{L}, \frac{\alpha^2}{p_T}, \left(\frac{N \alpha^2}{L V }\right)^{1/3} \},
    \end{equation}
    where $p_T :=\max \{0,J_{\theta_T}^C-J_{\theta}^{C^*}-\beta\}$.
    We then have
    %\vspace{-3mm}
    \begin{align}
        \mathbb{E}& \left[J_{\theta_T}^C \right] - J_{\theta}^{C^*} \leq  \beta  \nonumber\\
        &+ p_T \left(1-\sqrt{\frac{\zeta^3 L V}{4 \alpha^2 N}}\right)^T  + \sqrt{\frac{L V \zeta \alpha^2} N}.
    \end{align}
    Furthermore, to show that the objective $J_{\theta_T}^C$ reaches within an $\epsilon+\beta$-ball, with $\epsilon>0$, of global minima $J_{\theta}^{C^*}$ with step-size $\zeta=\frac{(\epsilon \alpha)^2 N}{4  V}$, i.e.  $\mathbb{E} \left[J_{\theta_T}^C \right] - J_{\theta}^{C^*} \leq \epsilon + \beta$, the samples complexity follows
    %\vspace{-3mm}
    \begin{equation}
        N T \leq \frac{16 \alpha^4 L V}{\epsilon^3} \log\frac{p_T}{\epsilon}.
    \end{equation}
\end{theorem}
Theorem \ref{thm:disc-samp-comp} provides sufficient conditions under which the \texttt{SafePG} algorithm converges to an $(\epsilon+\beta)$-neighborhood of the global constrained optimum, with sample complexity scaling as $\mathcal{O}(\frac{1}{\epsilon^3}\log\frac{1}{\epsilon})$.

\section{Case Study}\label{sec:exp}
We now validate the proposed safe RL framework on a safe quadrotor  environment where the objective is to safely navigate around an obstacle to reach the target. We demonstrate that \texttt{SafePG} learns a safe  policy that enables the quadrotor to reach its goal, in comparison to a projection-based safety-filter approach where the quadrotor remains safe but fails to converge to the goal.

\subsection{Setup}
We consider the safe quadrotor navigation problem in \cite{xu2018safe}, where we denote the positions of the quadrotor, the obstacle and the goal as $r = (r_x, r_y, r_z)$, $r_{obs}=(r^o_{x},r^o_{y},r^o_{z})$ and $r_{goal}$ respectively in the world frame. Under the small-angle approximation, the quadrotor dynamics can be modeled as a double integrator $\ddot{r} = u$, where $u \in \mathbb{R}^3$ is the acceleration in the $x-,y-,z-$axes. The actuator limits are defined by a fixed hyperrectangle $H := \{u \in \mathbb{R}^3 \mid u_{\min} \leq u \leq u_{\max}\}$. To ensure collision-free navigation around an obstacle located at $r_{obs}$, we define the safe set $\mathcal{S} = \{r : h(r) \geq 0\}$, where the CBF is given by:
\begin{equation}
    h(r) = \left(\frac{\Delta r_x}{a}\right)^4 + \left(\frac{\Delta r_y}{b}\right)^4 + \left(\frac{\Delta r_z}{c}\right)^4- r_s,
\end{equation}
where $\Delta r = r - r_{obs}$ is the distance to the obstacle,  parameters $a,b > 0$ represent the shape of the obstacle, and $r_s > 0$ is the safety margin. Because the relative degree \cite{ames2019control,xiao2021high} of $h(r)$ with respect to $u$ is 2, we enforce forward invariance using an Exponential CBF condition \cite{xu2018safe}:
\begin{equation}\label{eq:ecbf-case}
    \ddot{h}(r) + K_2 \dot{h}(r) + K_1 h(r) \geq 0,
\end{equation}
where $K_1, K_2 > 0$ are design parameters. This condition can be rearranged into an affine constraint on the control input, $A_r u \leq b_r$, yielding the state-dependent safe control set $\mathcal{C}(r) = \{u : A_r u \leq b_r\}$. 
% In this paper, we consider navigation in the $x$ and $y$ dimensions, with $K_1=6, K_2=8$, and corresponding actuator constraints $H_2:= \{u \in \mathbb{R}^2 \mid u_{\min} \leq u \leq u_{\max}\}$.
{In this paper, we consider navigation in the $x$ and $y$ dimensions with corresponding actuator constraints $H_2:= \{u \in \mathbb{R}^2 \mid u_{\min} \leq u \leq u_{\max}\}$, and use $K_1=6,K_2=8$, for which $s^2+K_2s+K_1=(s+4-\sqrt{10})(s+4+\sqrt{10})$ is Hurwitz and overdamped.
% Increasing $K_1$ relaxes \eqref{eq:ecbf-case} when $h>0$ but has no direct effect at $h=0$; increasing $K_2$ tightens the admissible controls when $\dot h<0$ (approaching the obstacle) and relaxes them when $\dot h>0$ (receding). The pair therefore penalizes approach without inducing oscillatory barrier dynamics, while retaining a nonempty control box along the reported trajectories.
} We discretize the dynamics and safety condition with discretization period set to $0.1$ seconds. A more detailed experiment setup including all the hyperparameters for the quadrotor dynamics and CBF condition can be found in \cite{sharma2026risk}. %\chck{add discretization?}

\subsection{\texttt{SafePG} Implementation}
% \chck{cite our Github and OJCSYS for the parameters and implementation of the environment?} \vks{added}
We first parameterize the underlying deterministic policy $\mu_\theta$ as a PI controller:
\begin{equation}
    \mu_\theta(x_t) = K_P e_t + K_I \sum_{i=0}^{t-1} e_i,
\end{equation}
where $\theta = \begin{bmatrix}
    K_P' \quad
    K_I'
\end{bmatrix}'$, $e_t = r_{goal} - r_t$ is the distance of the quadrotor from the goal, and the hyperparameters are chosen as $\Sigma_0=100 I_2, M=200, N=200, \gamma = 0.95, \theta_0=\{0.1,0.1\}$.
% \begin{table}[htb]
% \centering\small
% \caption{Case-study and \texttt{SafePG} parameters.}\label{tab:params}
% \begin{tabular}{@{}ll@{}}\hline
% CBF/Discretization period & $(K_1,K_2,\Delta)=(6,8,0.1\,\mathrm{s})$\\
% MC/batch/horizon & $M=$, $B=$, $H=$\\
% Discount/step size & $\gamma=$, $\eta=$\\
% Exploration/seeds & $\Sigma=$, seeds $=$\\ \hline
% \end{tabular}
% \end{table}
For truncation of Gaussian policy $\pi_\theta(\cdot | x_t) \sim \mathcal{N}(\mu_\theta(x_t), \Sigma_t)$, we construct a maximal inner hyperrectangle $H_c(r_t) \subseteq \mathcal{C}(r_t) \cap H_2$, adapted from \cite{repo}, and then truncate the support of $\pi_\theta(\cdot | x_t)$ to this constraint set $H_c(r_t)$ to obtain $\pi_\theta^C(\cdot | x_t)$.
% To construct the truncated PI policy, we truncate the unconstrained policy $\pi_\theta(\cdot | x_t) \sim \mathcal{N}(\mu_\theta(x_t), \Sigma^2)$ onto the safe control set. The strictly safe, truncated policy $\pi^C_\theta(\cdot | x_t)$ is then constructed by truncating the support of $\pi_\theta$ to the safe control set.
We define the per-time step reward as $-\norm{e_t}$ and minimize the episodic cost $\sum_{i=t}^{t+N-1} \norm{e_t}$  to encourage the quadrotor to reach the goal while exploring safely with $\pi_\theta^C(\cdot | x_t)$.

\subsection{Benchmark and Results}
In order to highlight the importance of searching directly over a class of safe policies, we compare \texttt{SafePG} with a projection-based safety filter baseline. For the baseline, we first learn an unconstrained stochastic PI policy $\pi_\theta(\cdot | x_t)\sim \mathcal{N}(\mu_\theta(x_t), \Sigma_t)$ through a policy gradient approach and project it  onto the constraint set $H_c(r(t))$ to obtain a safe control value $u_t^{project}$ by solving a Quadratic Program (QP) at each time step.
% --- FIGURE PLACEHOLDER ---
% \begin{figure}[t]
%     \centering
%     \vspace{4.5cm} % Placeholder space for the trajectory comparison plot
%     % \includegraphics[width=\linewidth]{figures/trajectory_comparison.pdf}
%     \caption{Comparison of quadrotor trajectories. The proposed \texttt{safePG} algorithm (blue) successfully learns to navigate around the unsafe zone to reach the goal. In contrast, the projection-based safety filter baseline (orange) gets trapped near the obstacle boundary and fails to reach the target.}
%     \label{fig:trajectory_comparison}
% \end{figure}
%
% \begin{figure}[htbp]
%     \centering
%     %\vspace{4cm} % Placeholder space for the trajectory plot
%     \includegraphics[width=0.87\linewidth]{safepg_learning_trajectories.png}
%     \caption{Performance of the \texttt{safePG} algorithm where the quadrotor reaches the goal while maintaining safety at all times during exploration.}
%     \label{fig:trajectory}
% \end{figure}
% \begin{figure}[htbp]\label{fig:projection_trajectory}
%     \centering
%     \includegraphics[width=0.9\linewidth]{learning_progression.png}
%     \caption{Performance of the projection-based safety filter where the quadrotor maintains safety but fails to converge to the goal.}
% \label{fig:projection_full_results}
% \end{figure}
%\vspace{-4cm}
\begin{figure}[t]
    \centering
    \includegraphics[width=0.70\linewidth]
        {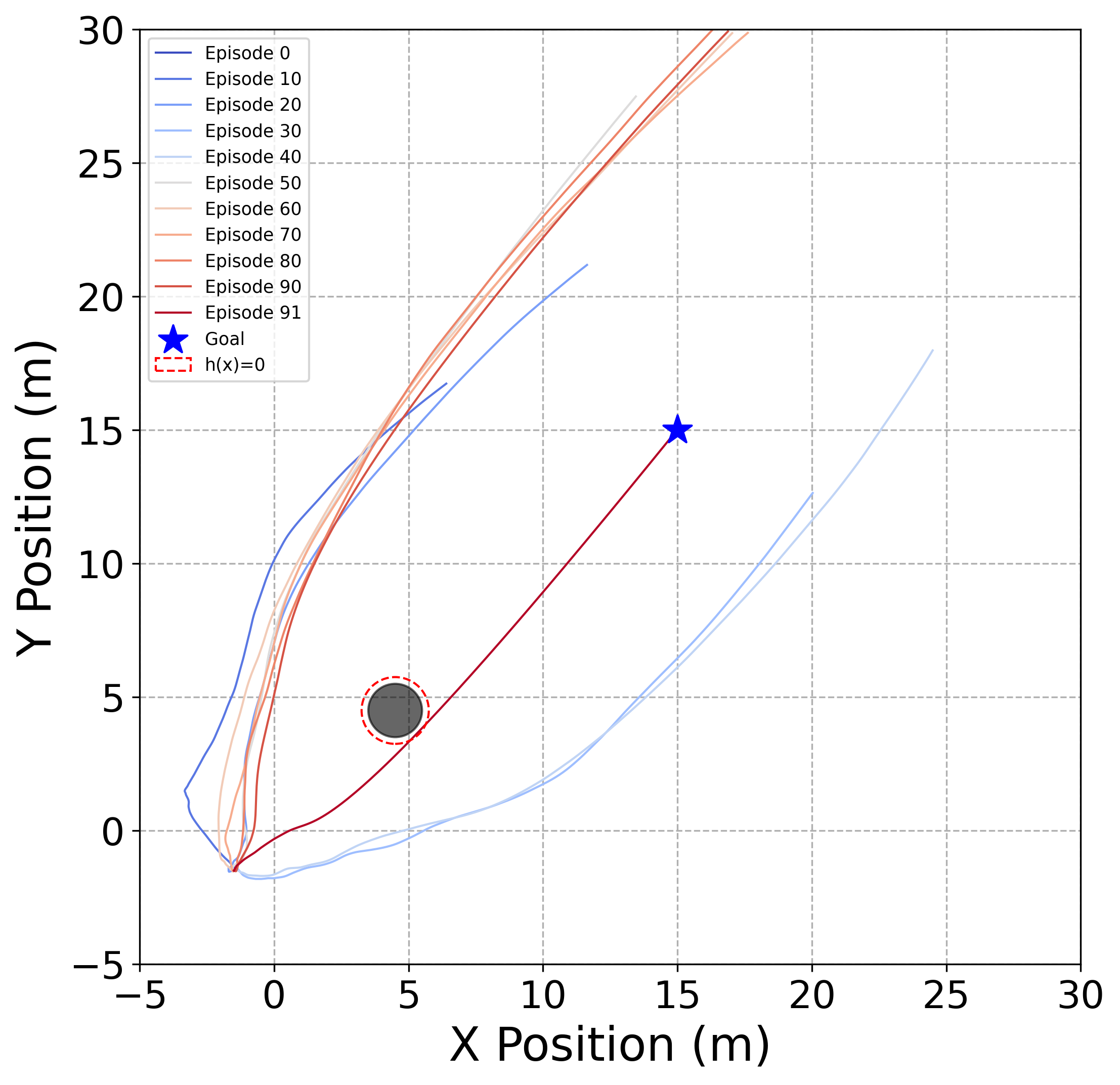}
    \caption{\texttt{SafePG} training trajectories showing
    obstacle avoidance and goal attainment.}
    \label{fig:trajectory}
\end{figure}

\begin{figure}[t]
    \centering
    \includegraphics[width=0.70\linewidth]
        {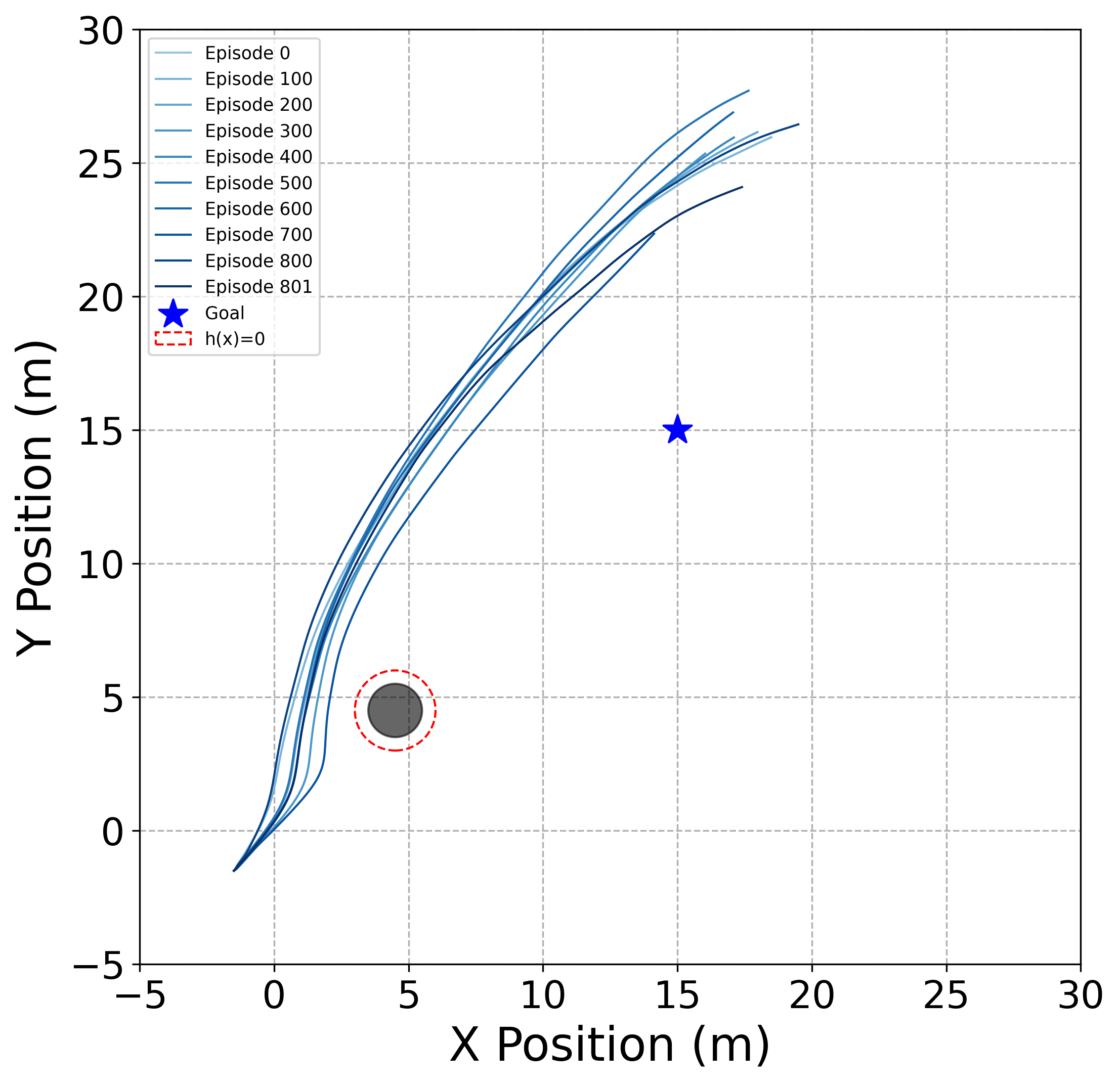}
    \caption{Projection-based training trajectories:
    the goal is not reached in the displayed episodes.}
    \label{fig:projection_full_results}
\end{figure}
As illustrated in Figure \ref{fig:trajectory}, the \texttt{SafePG} agent successfully learns to navigate around the obstacle to reach the goal while strictly maintaining safety{, even when we varied the exploration noise variance $\Sigma_0$ from $10 I_2$ to $100 I_2$.} Conversely, while the projection-based safety filter prevents collisions, the agent fails to reach the goal even after 800 episodes, as shown in Figure \ref{fig:projection_full_results}. 
% \begin{figure}[htbp]
%     \centering
%     %\vspace{4cm} % Placeholder space for the projection trajectory plot
%     \includegraphics[width=\linewidth]{learning_progression.png}
%     \caption{Quadrotor trajectories generated by the projection-based safety filter. While safety is maintained, the agent fails to reach the goal. \chck{add reward curves too?}}
%     \label{fig:projection_trajectory}
% \end{figure}
\section{Conclusion}
We developed a model-free safe RL framework which satisfies hard pointwise safety constraints at all times, and provided conditions for convergence and global optimality of the proposed PG algorithm. We validated this design on a safe quadrotor navigation problem, and showed that quadrotor reaches its target while avoiding an obstacle in its path. Future directions include integrating with data-driven or Gaussian process learning to jointly learn the safe set or safety specifications along with the policy parameters while preserving formal safety and convergence guarantees{, and development of scalable policy gradient framework for high-order or multi-agent systems.}

\balance
\bibliographystyle{IEEEtran}
\bibliography{references} 

\appendix
% \subsection{Additional Calculations}\label{app-subsec:add_detail}
% \textbf{Nothing to write yet}
% \subsection{Proofs}\label{app-subsec:Proofs}
% \begin{proof}[Proof of Lemma \ref{lem:piC_diff}]
% \end{proof}
\begin{proof}[\textbf{Lemma \ref{lem:score}}]
    Recall from the definition of $\pi_{\theta}^C$  in \eqref{eq:piC_definition} that, for $u \in C(x_t)$,
    \begin{align*}
        \log \pi_{\theta}^C(u_t|x_t) &= \log \frac{\pi_{\theta}(u_t|x_t)}{\pi_{\theta}(C(x_t)|x)},\\
        &=\log \pi_{\theta}(u_t|x_t)- \log \pi_{\theta}(C(x_t)|x)
    \end{align*}
    Taking the operator $\nabla_{\theta}$ on both sides, we get
    \begin{align}
        &\nabla_{\theta} \log \pi_{\theta}^C(u_t|x_t) %\nonumber\\
        % &= \nabla_{\theta} \log \pi_{\theta}(u_t|x_t) - \nabla_{\theta} \log \pi_{\theta}(C(x_t)|x),\nonumber\\
        %&
        = \frac{\nabla_{\theta} \pi_{\theta}(u_t|x_t)}{\pi_{\theta}(u_t|x_t)} - \frac{\nabla_{\theta} \pi_{\theta}(C(x_t)|x)}{\pi_{\theta}(C(x_t)|x)} \label{eq:c-grad}.
    \end{align}
       With the policy $\pi_{\theta}$ from \eqref{eq:unsafe_pi}, given by
        $$\pi_{\theta}(u_t | x_t)= (2\pi)^{-\frac{m}{2}} det(\Sigma_t)^{-\frac{1}{2}}e^{-\frac{1}{2}[u_t-\mu_\theta(x_t)]^{\top} \Sigma_t^{-1} [u_t-\mu_\theta(x_t)]},$$
    % \begin{align}\label{eq:pi-density}
    %     &\pi_{\theta}(u_t | x_t) \nonumber\\
    %     &\quad = (2\pi)^{-\frac{m}{2}} det(\Sigma_t)^{-\frac{1}{2}}e^{-\frac{1}{2}[u_t-\mu_\theta(x_t)]^{\top} \Sigma_t^{-1} [u_t-\mu_\theta(x_t)]}.
    % \end{align}
    we can write the  gradient 
    $$\nabla_{\theta} \pi_{\theta}(u_t|x_t) = \pi_{\theta}(u_t|x_t) (\frac{\partial\mu_\theta(x_t)}{\partial \theta})^\top (\Sigma_t^{-1})^\top (u_t - \mu_\theta(x_t)).$$
    % \begin{equation*}
    % \small
    %         \nabla_{\theta} \pi_{\theta}(u_t|x_t) = \pi_{\theta}(u_t|x_t) (\frac{\partial\mu_\theta(x_t)}{\partial \theta})^\top (\Sigma_t^{-1})^\top (u_t - \mu_\theta(x_t)).
    % \end{equation*}
    % From Lemma \ref{lem:piC_diff}, we know that $\nabla_{\theta} \pi_{\theta}(u_t|x_t) = \pi_{\theta}(u_t|x_t) \left[(\frac{\partial\mu_\theta(x_t)}{\partial \theta})^\top (\Sigma_t^{-1})^\top (u_t - \mu_\theta(x_t))\right]$.
    Then, the first term  in \eqref{eq:c-grad} can be written as $$\frac{\nabla_{\theta} \pi_{\theta}(u_t|x_t)}{\pi_{\theta}(u_t|x_t)} = (\frac{\partial\mu_\theta(x_t)}{\partial \theta})^\top (\Sigma_t^{-1})^\top (u_t - \mu_\theta(x_t)).$$
    % \begin{equation}\label{eq:a-grad}
    %     \frac{\nabla_{\theta} \pi_{\theta}(u_t|x_t)}{\pi_{\theta}(u_t|x_t)} = (\frac{\partial\mu_\theta(x_t)}{\partial \theta})^\top (\Sigma_t^{-1})^\top (u_t - \mu_\theta(x_t)).
    % \end{equation}
    For the second term $\frac{\nabla_{\theta} \pi_{\theta}(C(x_t)|x_t)}{\pi_{\theta}(C(x_t)|x_t)}$ in \eqref{eq:c-grad}, we begin by recalling the definition of $\pi_{\theta}(C(x_t)|x_t)$ and applying the gradient operator $\nabla_\theta$ as
    % \begin{equation*}
    %     \pi_{\theta}(C(x_t) | x_t) = \int_{C(x_t)} \pi_{\theta}(u | x_t) du.
    % \end{equation*}
    % Applying the $\nabla_{\theta}$ operator on both sides, we get
    \begin{equation*}
       \nabla_{\theta} \pi_{\theta}(C(x_t) | x_t) = \nabla_{\theta} \int_{C(x_t)} \pi_{\theta}(u | x_t) du.
    \end{equation*}
    Note that since the gradient is with respect to $\theta$, we can move the $\nabla_{\theta}$ operator inside the integral as
    \begin{align*}
       &\nabla_{\theta} \pi_{\theta}(C(x_t) | x_t) \\
       &\quad = \int_{C(x_t)} \nabla_{\theta}  \pi_{\theta}(u | x_t) du,\\
       % &=\int_{C(x_t)} \pi_{\theta}(u|x_t) \left[(\frac{\partial\mu_\theta(x_t)}{\partial \theta})^\top (\Sigma_t^{-1})^\top (u - \mu_\theta(x_t))\right] du\\
       &\quad =\int_{C(x_t)} (\frac{\partial\mu_\theta(x_t)}{\partial \theta})^\top (\Sigma_t^{-1})^\top u \pi_{\theta}(u|x_t) du\\
       &\quad \quad - \int_{C(x_t)} (\frac{\partial\mu_\theta(x_t)}{\partial \theta})^\top (\Sigma_t^{-1})^\top \mu_\theta(x_t) \pi_{\theta}(u|x_t) du\\
       &\quad =(\frac{\partial\mu_\theta(x_t)}{\partial \theta})^\top (\Sigma_t^{-1})^\top\int_{C(x_t)}  u_t \pi_{\theta}(u|x_t) du\\
       &\quad \quad - (\frac{\partial\mu_\theta(x_t)}{\partial \theta})^\top (\Sigma_t^{-1})^\top \mu_\theta(x_t) \int_{C(x_t)}  \pi_{\theta}(u|x_t) du.
    \end{align*}
    Therefore we can write 
    \begin{align*}
    &\frac{\nabla_{\theta} \pi_{\theta}(C(x_t)|x_t)}{\pi_{\theta}(C(x_t)|x_t)} \nonumber\\
           &=(\frac{\partial\mu_\theta(x_t)}{\partial \theta})^\top (\Sigma_t^{-1})^\top \frac{\int_{C(x_t)}  u \pi_{\theta}(u|x_t) du}{\pi_{\theta}(C(x_t)|x)}\\
       &\quad - (\frac{\partial\mu_\theta(x_t)}{\partial \theta})^\top (\Sigma_t^{-1})^\top \mu_\theta(x_t) \frac{\int_{C(x_t)}  \pi_{\theta}(u|x_t) du}{\pi_{\theta}(C(x_t)|x_t)}.
    \end{align*}
    From the definition $\pi_{\theta}(C(x_t)|x_t) = \int_{C(x_t)}  \pi_{\theta}(u|x_t) du$ and $m_C(x_t)$ from \eqref{eq:m_H}, we can finally write 
    \begin{align*}
        \nabla_{\theta} \log \pi_{\theta}^C(u_t|x_t)
        =(\frac{\partial\mu_\theta(x_t)}{\partial \theta})^\top (\Sigma_t^{-1})^\top \left[u_t - m_C(x_t)\right ].
    \end{align*}
    % \begin{align*}
    %     \frac{\nabla_{\theta} \pi_{\theta}(C(x_t)|x)}{\pi_{\theta}(C(x_t)|x)}=
    %     % &(\frac{\partial\mu_\theta(x_t)}{\partial \theta})^\top (\Sigma_t^{-1})^\top m_C(x_t) -\nonumber\\
    %     % &(\frac{\partial\mu_\theta(x_t)}{\partial \theta})^\top (\Sigma_t^{-1})^\top \mu_\theta(x_t),\nonumber\\
    %     % =&
    %     (\frac{\partial\mu_\theta(x_t)}{\partial \theta})^\top (\Sigma_t^{-1})^\top (\mu_\theta(x_t) -  m_C(x_t)).\label{eq:b-grad}
    % \end{align*}
    % Then, we rewrite $\nabla_{\theta} \log \pi_{\theta}^C(u_t|x_t)$ from \eqref{eq:c-grad} as
\end{proof}
\begin{proof}[\textbf{Theorem \ref{thm:disc-samp-comp}}]
% We begin by defining a sequence of $\varepsilon$-balls, each reachable from the previous element of the sequence, that leads from $x_0$ to $\mathcal{F}${, for any given $x_0 \sim \mc{D}$. } We then show that the head of the Markov chain $\{x_k\}_{k \in \mathbb{N}}$ lies inside this sequence with a positive probability. 
For a given $\varepsilon > 0$, suppose the finite sequence $\{ y_0, y_1, \ldots, y_N \} \subset \mathcal{S}$ is such that $B_{\varepsilon}(y_{k+1}) \subset R(B_{\varepsilon}(y_{k}))$, for $k = 1, \ldots, N-1$, and $B_{\varepsilon}(y_N) \cap \mathcal{F} \neq \emptyset$. For a given $\theta$, let $\{x_k\}_{k \in \mathbb{N}}$ be the Markov chain induced on $\mathcal{S}$ by $\pi^C_{\theta}$ such that $x_0 = y_0$. We then show that the finite trajectory $( x_0, x_1, \ldots, x_N )$ from the Markov chain is contained within the set $\{ y_0 \} \times B_{\varepsilon}(y_1) \times \ldots \times B_{\varepsilon}(y_N)$ with strictly positive probability, which will imply that $\{x_k\}_{k \in \mathbb{N}}$ enters $\mathcal{F}$ with strictly positive probability.

For each $k = 1, \ldots, N$, consider the probability measure $\nu_{\theta}$ defined as
\begin{equation*}
    \nu_{\theta}(\mc{A}) = P(x_k \in S \ | \ x_{k-1}) = \int_{\mc{M}^{-1}_{x_{k-1}}(\mc{A})} \pi^C_{\theta}(a | x_{k-1}) \ da,
\end{equation*}
for any $\mc{T}_{\#} \ell$-measurable subset $\mc{A}$ of $\mathcal{S}$. Note that $\nu_{\theta}$ is absolutely continuous with respect to $\mc{T}_{\#} \ell$, written $\nu_{\theta} \ll \mc{T}_{\#} \ell$, since $\mc{T}_{\#} \ell(\mc{A}) > 0$ if and only if {$\ell(\mc{M}^{-1}_{x_{k-1}}(\mc{A})) > 0$}, by Assumption \ref{asum:volume_preservation}. Then, the Lebesgue-Radon-Nikodym Theorem from \cite[\S3.2]{folland1999real} implies that there exists a $\mc{T}_{\#} \ell$-integrable function $f_{\theta} : \mathcal{S} \rightarrow \mathbb{R}$, called the Radon-Nikodym derivative of $\nu_{\theta}$, such that $\nu_{\theta}(\mc{A}) = \int_S f_{\theta}(x) dx$ (see \cite{folland1999real} for details). To make the link between $f_{\theta}$ and $\nu_{\theta}$ perfectly clear, let us write $$f_{\theta}(x) = \int_{\mc{M}^{-1}_{x_{k-1}}(x)} \pi^C_{\theta}(a | x_{k-1}) \ da,$$ $$\nu_{\theta}(\mc{A}) = \int_S \int_{\mc{M}^{-1}_{x_{k-1}}(x)} \pi^C_{\theta}(a | x_{k-1}) \ da \ dx.$$
% \begin{align*}
%     f_{\theta}(x) &= \int_{\mc{M}^{-1}_{x_{k-1}}(x)} \pi^C_{\theta}(a | x_{k-1}) \ da, \\
%     \nu_{\theta}(\mc{A}) &= \int_S \int_{\mc{M}^{-1}_{x_{k-1}}(x)} \pi^C_{\theta}(a | x_{k-1}) \ da \ dx.
% \end{align*}

By Assumptions \ref{asum:volume_preservation} and \ref{asum:positive_probability}, we also have $\mc{T}_{\#} \ell \ll \nu_{\theta}$. Since both $\nu_{\theta} \ll \mc{T}_{\#} \ell$ and $\mc{T}_{\#} \ell \ll \nu_{\theta}$, the two measures are said to be {\em equivalent}, meaning that they agree on which sets have measure zero. Since $\mc{T}_{\#} \ell$ and $\nu_{\theta}$ are equivalent, a standard result from real analysis allows us to take the Radon-Nikodym derivative $f_{\theta}$ to be strictly positive $\mc{T}_{\#} \ell$-almost everywhere.
As a first consequence, notice that
\begin{align}
    P \Big( &(x_0, x_1, x_2) \in \{y_0\} \times B_{\varepsilon}(y_1) \times B_{\varepsilon}(y_2) \ | \ x_0 = y_0 \Big) \nonumber \\
    &= P \Big( (x_1, x_2) \in B_{\varepsilon}(y_1) \times B_{\varepsilon}(y_2) \ | \ x_0 = y_0 \Big) \nonumber \\
    &\quad \times P(x_0 = y_0) \nonumber \\
    &= P \Big( (x_1, x_2) \in B_{\varepsilon}(y_1) \times B_{\varepsilon}(y_2) \ | \ x_0 = y_0 \Big) \nonumber \\
    &= \int_{B_{\varepsilon}(y_1)} \int_{T^{-1}_{x_1} (B_{\varepsilon}(y_2))} \pi^C_{\theta} (a_1 | x_1) f_1(x_1) \ da_1 \ dx_1 \label{eqn:positive1} \\
    &= \int_{B_{\varepsilon}(y_1)} \int_{T^{-1}_{x_1} (B_{\varepsilon}(y_2)} \pi^C_{\theta} (a_1 | x_1) \nonumber \\
    &\quad \times \left[ \int_{T^{-1}_{x_0}(x_1)} \pi^C_{\theta} (a_0 | x_0) \ da_0 \right] \ da_1 \ dx_1. \nonumber
\end{align}

Given Assumption \ref{asum:positive_probability},  \eqref{eqn:positive1} is strictly positive, since $f_1$ is strictly positive almost everywhere and the integrals are taken over sets of positive volume. Then, we have

\begin{align}
    P \Big( &(x_1, \ldots, x_{N-1}, x_N) \in B_{\varepsilon}(y_1) \times \ldots \nonumber \\
    &\times B_{\varepsilon}(y_{N-1}) \times (B_{\varepsilon}(y_N) \cap \mc{F}) \ | \ x_0 = y_0 \Big) \label{eqn:positive2} \\
    = &\int_{B_{\varepsilon}(y_1)} \int_{T^{-1}_{x_1}(B_{\varepsilon}(y_2))} \pi^C_{\theta}(a_1 | x_1) \nonumber \\
    &\times \int_{B_{\varepsilon}(y_2)} \int_{T^{-1}_{x_2}(B_{\varepsilon}(y_3))} \pi^C_{\theta}(a_2 | x_2) \times \ldots  \nonumber \\
    &\times \int_{B_{\varepsilon}(y_{N-1})} \int_{T^{-1}_{x_{N-1}}(B_{\varepsilon}(y_N) \cap \mc{F})} \pi^C_{\theta}(a_{N-1} | x_{N-1}) \nonumber \\
    &\times f_{N-1}(x_{N-1}) \ da_{N-1} \ dx_{N-1} \times \ldots \nonumber \\
    &\times f_2(x_2) \ da_2 \ dx_2 \cdot f_1(x_1) \ da_1 \ dx_1. \nonumber
\end{align}

In the innermost integral, we have $\mc{T}_{\#} \ell(B_{\varepsilon}(y_N) \cap \mc{F}) > 0$, since both sets are open and their intersection is non-empty by hypothesis. Finally, given Assumption \ref{asum:positive_probability}, we have that \eqref{eqn:positive2} is strictly positive, since all integrals are taken over sets of positive volume and $f_i$ is strictly positive almost everywhere, for each $i \in \{1, \ldots N-1\}$. This implies that every Markov chain $\{x_k\}_{k \in \mathbb{N}}${, with $x_0 \sim \mc{D}$,} induced by $\pi_\theta^C$ has a positive probability. Consequently, we can conclude that every $\pi_\theta^C$-induced Markov chain is $\mc{T}_{\#} \ell-$irreducible.
\end{proof}

\vspace{0.75em}

\begin{proof}[\textbf{Corollary \ref{corol:well-defined}}]
    From Theorem \ref{thm:irreducibility}, we know that the $\pi_\theta^C-$induced Markov chain is irreducible, which implies that there every induced Markov chain has a unique occupancy or probability measure. Therefore, we can conclude that the expectation operator $\mathbb{E}_{\pi_\theta^C}$ in the safe RL objective in \eqref{eq:obj_fn} is well-defined.
\end{proof}

\vspace{0.75em}

\begin{proof}[\textbf{Theorem \ref{thm:disc-samp-comp}}]
    % Using the fact that $x \leq \|x\|$ and Corollary \ref{cor:app:weak-wgd}, we can write
    % \begin{equation}\label{eq:wgd-1}
    %     J^C_{\theta}(\theta) - J^C_{\theta}(\theta^*) \leq \alpha \norm{ \nabla_{\theta} J^C_{\theta}(\theta) } + \beta.
    % \end{equation}
    Define $\bar{J}^{C^*}_{\theta} := J^{C^*}_{\theta} + \beta$. Then, from \eqref{eq:wgd}, we have the following: 
    \begin{equation}\label{eq:wgd-2}
        \frac{1}{\alpha} \max\{0, J_{\theta}^C - \bar{J}^{C^*}_{\theta}\} \leq \norm{ \nabla_{\theta} J^C_{\theta}(\theta) }.
    \end{equation}
    We write $J_\theta^C$ at $\theta_t$, at time step $t$, as $J_\theta^C(\theta_t)$. Then, we note that, for time-step $t \geq 0$,
    \begin{equation}\label{eq:wgd-3}
        \begin{aligned}
            J^C_{\theta}(\theta_{t+1}) - \bar{J}^{C^*}_{\theta} =& J^C_{\theta}(\theta_{t+1}) - J^C_{\theta}(\theta_t) \\
            &+ J^C_{\theta}(\theta_t) - \bar{J}^{C^*}_{\theta}.
        \end{aligned}
    \end{equation}
    Using the Lipschitzness condition \eqref{eq:lip-grad} and the second-order Taylor expansion gives
    \begin{equation}\label{eq:wgd-4}
        \begin{aligned}
            J^C_{\theta}(\theta_{t+1}) - J^C_{\theta}(\theta_t) \leq& \langle  \nabla_{\theta} J_{\theta}^C(\theta_t), \theta_{t+1} -\theta_t \rangle \\
            &+ \frac{L}{2} \|\theta_{t+1} -\theta_t \|^2_2,
        \end{aligned}   
    \end{equation}
    and we therefore have that 
    \begin{equation}\label{eq:wgd-5}
        \begin{aligned}
            J^C_{\theta}(\theta_{t+1}) - \bar{J}^{C^*}_{\theta} \leq& J^C_{\theta}(\theta_t) - \bar{J}^{C^*}_{\theta} \\
            &+ \langle \nabla_{\theta} J_{\theta}^C(\theta_t), \theta_{t+1} -\theta_t \rangle \\
            &+ \frac{L}{2} \|\theta_{t+1} -\theta_t \|^2_2.
        \end{aligned} 
    \end{equation}

    Recalling the update rule $\theta_{t+1} = \theta_{t} - \eta_t \widehat{ \nabla_{\theta} } J_{\theta}^C(\theta_t)$, by the unbiasedness of the gradient estimator we have
    \begin{equation}\label{eq:wgd-6}
         \begin{aligned}
            J^C_{\theta}(\theta_{t+1}) - \bar{J}^{C^*}_{\theta} \leq& J^C_{\theta}(\theta_t) - \bar{J}^{C^*}_{\theta} \\
            &- \eta_t \langle  \nabla_{\theta} J_{\theta}^C(\theta_t) , \widehat{ \nabla_{\theta} } J_{\theta}^C(\theta_t)  \rangle \\
            &+\frac{L}{2} \eta_t^2 \norm{\widehat{ \nabla_{\theta} } J_{\theta}^C(\theta_t) }^2.
        \end{aligned}  
    \end{equation}
    
    Let $\mc{F}_t = \sigma(\theta_{\theta}, \eta_{\theta}; 0 \leq k \leq t)$ denote the $\sigma$-algebra generated by all randomness in the system up until time $t$, and let $\mathbb{E}_t [ \cdot ] = \mathbb{E} [ \cdot | \mc{F}_{t-1} ]$.
    Taking expectations over the foregoing yields
    \begin{equation}\label{eq:wgd-7}
         \begin{aligned}
            &\mathbb{E}_t\Big[J^C_{\theta}(\theta_{t+1}) - \bar{J}^{C^*}_{\theta} \Big] \\
            &\leq  \mathbb{E}_t \Big[J^C_{\theta}(\theta_t) - \bar{J}^{C^*}_{\theta} - \eta_t \langle  \nabla_{\theta} J_{\theta}^C(\theta_t) , \widehat{ \nabla_{\theta} } J_{\theta}^C(\theta_t) \rangle \\
            &\quad +\frac{L}{2} \eta_t^2 \norm{\widehat{ \nabla_{\theta} } J_{\theta}^C(\theta_t) }^2 \Big].
        \end{aligned}  
    \end{equation}   
    
    Given $\mathbb{E}_t [\widehat{ \nabla_{\theta} } J_{\theta}^C(\theta_t) ] = \nabla_{\theta} J_{\theta}^C(\theta_t)$ from Assumption \ref{asum:zero-bdd-grad}, and recalling the definition of the variance of a random variable, we can write the following
    \begin{equation}\label{eq:wgd-8}
         \begin{aligned}
            &\mathbb{E}_t \Big[J^C_{\theta}(\theta_{t+1}) - \bar{J}^{C^*}_{\theta}  \Big]\\
            &\leq  J^C_{\theta}(\theta_t) - \bar{J}^{C^*}_{\theta} 
            - \eta_t \langle \nabla_{\theta} J_{\theta}^C(\theta_t), \nabla_{\theta} J_{\theta}^C(\theta_t) \rangle \\
            &\quad +\frac{L}{2} \eta_t^2 Var(\widehat{\nabla_{\theta} } J_{\theta}^C(\theta_t) ) +\frac{L}{2} \eta_t^2\Big(\mathbb{E}_t\Big[\norm{\widehat{\nabla_{\theta} } J_{\theta}^C(\theta_t)} \Big]\Big)^2 \\
            &=  J^C_{\theta}(\theta_t) - \bar{J}^{C^*}_{\theta} - \eta_t  \norm{\nabla_{\theta} J_{\theta}^C(\theta_t)}^2 \\
            &\quad +\frac{L}{2} \eta_t^2 Var(\widehat{\nabla_{\theta} } J_{\theta}^C(\theta_t) ) +\frac{L}{2} \eta_t^2\norm{{\nabla_{\theta} } J_{\theta}^C(\theta_t)}^2 \\
            &=  J^C_{\theta}(\theta_t) - \bar{J}^{C^*}_{\theta} - \eta_t \Big(1 - \frac{L}{2} \eta_t \Big) \norm{{\nabla_{\theta} } J_{\theta}^C(\theta_t)}^2 \\
            &\quad + \frac{L}{2} \eta_t^2 Var(\widehat{\nabla_{\theta} } J_{\theta}^C(\theta_t)).
        \end{aligned}  
    \end{equation}   
    Using Assumption \ref{asum:zero-bdd-grad}, we finally obtain
    \begin{equation}\label{eq:wgd-10}
         \begin{aligned}
            &\mathbb{E}_t \Big[J^C_{\theta}(\theta_{t+1}) - \bar{J}^{C^*}_{\theta} \Big]\\
            &\leq  J^C_{\theta}(\theta_t) - \bar{J}^{C^*}_{\theta} - \eta_t \Big(1 - \frac{L}{2} \eta_t \Big) \norm{{\nabla_{\theta} } J_{\theta}^C(\theta_t)}^2 \\
            &\quad +\frac{L}{2} \eta_t^2 \frac{V}{N}.
        \end{aligned}  
    \end{equation}
     % For a sequence of Lipschitz constants $\{l(K_i)\}$, suppose that the largest Lipschitz constant is given by $L$.
     Further, suppose that we pick a step size $\eta_t \leq \frac{1}{L}$; then it follows that $(1-\frac{L}{2}\eta_t \geq \frac{1}{2})$. We can then re-write the above inequality as follows
    \begin{equation}\label{eq:wgd-11}
         \begin{aligned}
            &\mathbb{E}_t \Big[J^C_{\theta}(\theta_{t+1}) - \bar{J}^{C^*}_{\theta} \Big]\\
            &\leq  J^C_{\theta}(\theta_t) - \bar{J}^{C^*}_{\theta} - \frac{\eta_t } {2}\norm{{\nabla_{\theta}} J_{\theta}^C(\theta_t)}^2  +\frac{L}{2} \eta_t^2 \frac{V}{N}.
        \end{aligned}   
    \end{equation}
    
    At this point, we can now follow the proof of Theorem $6.1$ in \cite{montenegro2024learning}  with $L_2 = L$ to obtain the bound
    \begin{align}
            &\mathbb{E} \big[  J_{\theta}^C( \theta_T) \big] - J_{\theta}^{C^*} \nonumber \\
            &\leq \beta + p_T \left(1 - \sqrt{\frac{\eta^3 L V}{4 \alpha^2 N}} \right)^T + \sqrt{\frac{L V \eta \alpha^2}{N}}.
    \end{align}    
The remainder of the proof now follows in a straightforward manner from the proof of \cite[Theorem 6.1]{montenegro2024learning}.
\end{proof}

\end{document}